\documentclass[sigconf]{acmart}
\AtBeginDocument{%
  \providecommand\BibTeX{{%
    \normalfont B\kern-0.5em{\scshape i\kern-0.25em b}\kern-0.8em\TeX}}}

\copyrightyear{2022} 
\acmYear{2022} 
\setcopyright{rightsretained}
\acmConference[e-Energy '22]{The Thirteenth ACM International
Conference on Future Energy Systems}{June 28--July 1, 2022}{Virtual Event, USA}
\acmBooktitle{The Thirteenth ACM International Conference on Future Energy Systems (e-Energy '22), June 28--July 1, 2022, Virtual Event, USA}
\acmPrice{}
\acmDOI{10.1145/3538637.3538853}
\acmISBN{978-1-4503-9397-3/22/06}



\usepackage{algorithm}  \usepackage{amsmath}  \usepackage{amsthm}  
\usepackage{enumitem}  \usepackage{graphicx}  \usepackage{bbm}
\usepackage{caption}
\usepackage{subcaption}
\usepackage{cleveref}  \usepackage{siunitx}  \usepackage{xcolor}  \usepackage{mathtools}  
\newcommand{\N}{\mathbb{N}}  % natural numbers
\newcommand{\R}{\mathbb{R}}  % real numbers
\newcommand{\X}{\mathcal{X}}  % uncertainty set
\newcommand{\Z}{\mathbb{Z}}  % integers
\newcommand{\Sym}{\mathbb{S}}  % symmetric, real matrices
\newcommand{\norm}[1]{\left\lVert#1\right\rVert}  % norm
\newcommand{\one}{\mathbf{1}}  % ones vector
\newcommand{\set}[1]{\left\{#1\right\}}  % set
\newcommand{\zero}{\mathbf{0}}  % ones vector

\DeclareMathOperator{\diam}{diam}  % diameter
\DeclareMathOperator{\var}{Var}  % variance

\newcommand{\SEL}{\texttt{SEL}}
\newcommand{\CTRL}{\texttt{$\Pi$}}
\newcommand{\im}{\mathbf{i}}
\newcommand{\vpar}{v^\text{par}}
\newcommand{\Vpar}{\mathcal{V}^\text{par}}
\newcommand{\vnom}{v^\text{nom}}

\newcommand{\vmin}{\underline{v}}
\newcommand{\vmax}{\overline{v}}
\newcommand{\qmin}{\underline{q}}
\newcommand{\qmax}{\overline{q}}

\newcommand{\qlims}{[\qmin, \qmax]}
\newcommand{\vlims}{[\vmin, \vmax]}
\newcommand{\etamax}{\overline\eta}

\newtheorem*{theorem*}{Theorem}
\newtheorem{theorem}{Theorem}
\newtheorem{lemma}{Lemma}
\newtheorem{definition}{Definition}
\newtheorem{assumption}{Assumption}
\newtheorem*{assumption*}{Assumption}

\crefname{lemma}{Lemma}{Lemmas}
\crefname{definition}{Definition}{Definitions}
\crefname{assumption}{Assumption}{Assumptions}
\crefname{remark}{Remark}{Remarks}

\begin{document}

\title[Robust Online Voltage Control]{Robust Online Voltage Control with an Unknown Grid Topology}

\author{Christopher Yeh}
\affiliation{%
  \institution{California Institute of Technology}
  \city{Pasadena}
  \state{CA}
  \country{USA}
}
\email{cyeh@caltech.edu}

\author{Jing Yu}
\affiliation{%
  \institution{California Institute of Technology}
  \city{Pasadena}
  \state{CA}
  \country{USA}
}
\email{jing@caltech.edu}

\author{Yuanyuan Shi}
\affiliation{%
  \institution{UC San Diego}
  \city{San Diego}
  \state{CA}
  \country{USA}
}
\email{yyshi@eng.ucsd.edu}

\author{Adam Wierman}
\affiliation{%
  \institution{California Institute of Technology}
  \city{Pasadena}
  \state{CA}
  \country{USA}
}
\email{adamw@caltech.edu}

\begin{abstract}
Voltage control generally requires accurate information about the grid's topology in order to guarantee network stability. However, accurate topology identification is a challenging problem for existing methods, especially as the grid is subject to increasingly frequent reconfiguration due to the adoption of renewable energy. Further, running existing control mechanisms with incorrect network information may lead to unstable control. In this work, we combine a nested convex body chasing algorithm with a robust predictive controller to achieve provably finite-time convergence to safe voltage limits in the online setting where the network topology is initially unknown. Specifically, the online controller does not know the true network topology and line parameters, but instead must learn them over time by narrowing down the set of network topologies and line parameters that are consistent with its observations and adjusting reactive power generation accordingly to keep voltages within desired safety limits. We demonstrate the effectiveness of the approach using a case study, which shows that in practical settings the controller is indeed able to narrow the set of consistent topologies quickly enough to make control decisions that ensure stability.
\end{abstract}

\maketitle

\section{Introduction}
Operators of electricity distribution grids must maintain voltages at each bus within certain operating limits, as deviations from such limits may damage electrical equipment and cause power outages \cite{regulator,haes2019survey}. This ``voltage control'' or ``voltage regulation'' problem has been well-studied in the literature, \textit{e.g.}, ~\cite{rebours2007survey,zhu2015fast,molzahn2017survey} and the references therein. Voltage control algorithms aim to guarantee grid stability while minimizing the costs associated with control inputs. Typically, the algorithms assume \textit{exact knowledge} of the underlying grid topology.

However, the exact grid topology and line parameters are often not known, and inexact knowledge may lead to stability problems for voltage control algorithms~\cite{park2018exact,li2019robust}. This problem is exacerbated by the increasing integration of distributed energy resources (DERs), such as photovoltaic and storage devices.  Especially in distribution grids, where DERs often do not belong to the electricity utility, the grid operator may lack up-to-date information about the grid topology~\cite{liao2015distribution,deka2017structure}. While a grid operator can install sensors to help identify the current network topology, unless such sensors are densely deployed, uncertainty about the topology remains; and so cost is prohibitive. Furthermore, parts of the grid may undergo frequent reconfiguration either due to load balancing~\cite{baran1989network} or unplanned maintenance. Thus, distribution grid operators cannot expect to operate with perfect topology information and the design of voltage control robust to unknown grid topology is crucial.

To date, only a limited number of voltage control mechanisms have been studied in the case when the grid topology is uncertain. One common design is to learn a voltage controller via deep reinforcement learning (DRL), \textit{e.g.},~\cite{gao2021consensus,sun2021two,xu2020optimal,duan2019deep,wang2020data} and references within. However, such approaches have neither performance nor voltage stability guarantees. Thus, they are not suitable for safety-critical infrastructure. Two recent works~\cite{shi2021stability,cui2021decentralized} propose a model-free DRL approach for voltage control with stability guarantees. The main tool being used in~\cite{shi2021stability,cui2021decentralized} is Lyapunov stability theory, from which a structural constraint for stable controllers is derived, and policy optimization with the constraint is performed. In contrast, our framework jointly learns the system model (consistent to data) and stable controller, in an online fashion.

Another common approach for voltage control when grid topology is uncertain, is to use a two-stage model-based approach: first, estimate the network topology, a.k.a., system identification, using structured neural networks; and second, apply an existing model predictive control with the identified model, \textit{e.g.}, ~\cite{chen2020data}. However, accurate topology identification is a challenging problem~\cite{sharon2012topology} and existing methods require voltage measurements over hundreds of time steps~\cite{liao2015distribution,deka2016estimating}, after which uncertainty still remains. This is problematic because, during the time where system identification is being performed, the network is not able to respond effectively to disturbances, since using an incorrect model may lead to unstable control. Given the timescale of grid topology changes in practice, a different approach is needed.

\subsection{Contributions} We propose a new approach for voltage control over an uncertain grid topology that does not perform system identification and voltage control separately.  Instead, our approach robustly learns to stabilize voltage within the desired limits directly, \emph{without any prior knowledge of the topology and without needing to precisely learn the topology}.   

Our approach takes ideas from online nested convex body chasing (CBC) \cite{argue2019nearly} and robust  predictive control and combines them using a new learning framework~\cite{ho2021online} to apply them to voltage control for the first time. Intuitively, we use a nested CBC algorithm in order to track the set of topologies that are consistent with the observed voltage measurements---as more measurements are taken the set of consistent topologies shrinks (and so the sets are nested).  As these measurements are taken, a form of robust predictive control is used for voltage control, where the robustness guarantee is used to ensure the uncertainty about the topology can be handled.  Our main result (\Cref{thm:main}) provides a finite error stability bound for the overall controller, which is summarized in \Cref{alg:robust_online_volt_control}. This represents the first voltage control algorithm that is provably robust to uncertainty about network topology. 

In addition to providing theoretical guarantees, we demonstrate the effectiveness of our proposed approach using a case study of a 56-bus distribution grid from the Southern California Edison (SCE) utility~\cite{farivar2012optimal}. In this setting, we give the controller no prior information about the topology of the grid, yet the controller quickly narrows down the set of topologies and line parameters that are consistent with its observations and adjusts reactive power generation to keep voltages within desired safety limits when faced with disturbance. In fact, our controller's performance nearly matches that of controllers which assume perfect knowledge of the topology, even when given only partial observations of bus voltages.

\subsection{Related Work}
The problem of voltage control has a long history with many important contributions~\cite{senjyu2008optimal,gao2010review,zhang2014optimal,li2014realtime, zhu2015fast,bolognani2013distributed,tang2019fast,qu2020optimal,chen2020data,wang2020data} (and the references within). Classic voltage regulation devices such as tap-changing transformers~\cite{senjyu2008optimal,gao2010review} are effective in dealing with
\emph{slow} voltage variations. However, with fast time variations introduced by renewables, a
growing body of literature has focused on inverter-based controllers that can quickly respond by adjusting their active and reactive power set-points~\cite{li2014realtime, zhu2015fast,bolognani2013distributed,tang2019fast,qu2020optimal}. Most of these works cast voltage control as an optimization problem and then propose different centralized or decentralized algorithms depending on the communication infrastructure. Critically, all of these voltage control methods assume that the underlying grid topology is \emph{known}. 

Some recent works consider voltage control with unknown network topology and parameters. These works either use a two-stage model-based approach of first performing system identification and then optimizing over the identified model~\cite{chen2020data}, or an entirely model-free to learn a controller via deep reinforcement learning~\cite{gao2021consensus,sun2021two,xu2020optimal,duan2019deep,wang2020data} which has no performance or voltage stability guarantees. In contrast, our work considers model-based approach that jointly learns the system model and a controller.

An important tool for voltage control is model predictive control (MPC), which has been investigated in a number of works, \emph{e.g.}, \cite{guo2019mpc,chen2020data,maharjan2021robust}.  Of particular relevance to this paper is work on robust MPC algorithms for voltage control such as~\cite{maharjan2021robust}. While many proposals have emerged, to this point provable bounds for robust MPC algorithms have typically been elusive.  A key part of our proposed framework is the development of a robust control algorithm for voltage control with a provable robustness guarantee.  This is summarized in \Cref{thm:oracle}. 

The standard approach for handling uncertainty about network topology is to estimate the topology using a form of system identification. There is a growing literature of such approaches, e.g., \cite{kekatos2014grid,liao2015distribution,deka2016estimating,deka2017structure,park2018exact,li2019robust}. A prominent approach is to use graphical models for topology reconstruction~\cite{deka2016estimating}, via maximum likelihood methods while enforcing other structural restrictions like low-rank and sparsity. Our approach in this paper is novel because system identification is not performed separately from control.  Instead of seeking to estimate the topology, the controller itself is learned directly without the intermediate step of system identification.
\section{Model}
\label{sec:volt_ctrl_background}

We study voltage control on an unknown grid topology.  We consider a radial power distribution circuit represented as a network $G = (\mathcal{N}, \mathcal{E})$, where $\mathcal{N}$ is the set of buses (nodes) and $\mathcal{E} \subset \mathcal{N} \times \mathcal{N}$ is the set of lines (edges). The buses are numbered $\mathcal{N} = \{0, 1, 2, \dotsc, n\}$, where bus 0 is the substation and other buses are branch buses. Let $\mathcal{C} \subseteq \mathcal{N}$ denote the subset of buses with controllable reactive power injection. Because the network is radial (tree-structured) and rooted at the substation (bus 0), there is a unique path $\mathcal{P}_i$ from bus 0 to any other bus $i$.

For each line $(i,j) \in \mathcal{E}$, its complex impedance is $r_{ij} + \im x_{ij}$, with real-valued line parameters $r_{ij}, x_{ij} > 0$ (units \si{\ohm}). Define the following matrices $R^\star, X^\star \in \R^{n \times n}$, which are computed from the network topology and line parameters. Assuming that $G$ is a connected graph (\textit{i.e.}, no bus is disconnected from the substation), then $R^\star, X^\star$ are positive definite and have strictly positive entries \cite{farivar2013equilibrium}.
\begin{gather}
    R^\star_{ij} := 2 \sum_{\mathclap{(h,k) \in \mathcal{P}_i \cap \mathcal{P}_j}} r_{hk},
    \quad
    X^\star_{ij} := 2 \sum_{\mathclap{(h,k) \in \mathcal{P}_i \cap \mathcal{P}_j}} x_{hk},
    \quad
    i,j \in \{1, \dotsc, n\}
\end{gather}
Let $v \in \R^n$ denote the squared voltage magnitude at the buses, excluding the substation. Let $p + \mathbf{i} q$ denote the complex power injection at the buses, where $p \in \R^n$ (units \si{\watt}) is the net active power injection, and $q \in \R^n$ (units \si{\var}) is the net reactive power injection. We assume that the active power injection is exogenous, but that reactive power at each bus can be decomposed as $q = q^c + q^e$, where $q^c$ is the ``controllable'' component and $q^e$ is the ``exogenous'' (\emph{i.e.}, uncontrollable) component.

Under the linearized Simplified DistFlow model~\cite{li2014realtime},
\begin{equation}\label{eq:simplified_distflow}
    v = R^\star p + X^\star q + v^0 \one_n = X^\star q^c + \vpar,
\end{equation}
where $\vpar = R^\star p + X^\star q^e + v^0 \one_n \in \R^n$ (``par'' stands for ``partial'') represents the effect of the exogenous quantities on voltage, and $v^0$ is a known, fixed constant representing the squared voltage magnitude at the substation.

We can model this as a discrete-time linear system
\begin{equation}\label{eq:vol_dyn1}
    v(t+1) = X^\star q^c(t) + \vpar(t).
\end{equation}
Substituting $u(t) = q^c(t) - q^c(t-1)$ (change in controllable reactive power injection) and $w(t) = \vpar(t) - \vpar(t-1)$ (change in exogenous noise) yields the linear dynamical system
\begin{align}\label{eq:vol_dyn2}
    v(t+1)
        &= v(t) + X^\star u(t) + w(t).
        \end{align}

The voltage control problem~\cite{farivar2012optimal} is to drive the squared voltage magnitudes of each bus from an initial state $v(1) \in \R^n$ into a given multi-dimensional interval $\vlims \subset \R^n$; it is possible that $v(1)$ does not start within the interval due to some large initial disturbance. For all $t \geq 2$, the voltage control algorithm aims to maintain $v(t)$ within $\vlims$, ideally keeping $v(t)$ as close as possible to a ``nominal'' value $\vnom \in \vlims$, typically $\vnom = (\vmin + \vmax)/2$. The cost for deviating from $\vnom$ is measured by $   \norm{v(t) - \vnom}_{P_v}^2$ for some positive semidefinite matrix $P_v$, where $\norm{x}_A^2 := x^\top A x$.

At each time step, buses may change their reactive power injection in order to regulate the voltage close to $\vnom$. The reactive power injection must remain within a given bound $q^c(t) \in \qlims \subset \R^n$, and we assume $q^c(0)$ indeed starts within $\qlims$. Because buses not in $\mathcal{C}$ do not have any ability to control the reactive power injection: $\forall i \not\in \mathcal{C}.\ \qmin_i = \qmax_i = 0$. In our model, we do not place any hard ``ramp constraints'' on $u(t)$. However, we impose a quadratic ramping cost $\norm{u(t)}_{P_u}^2$ where $P_u$ is a positive semidefinite matrix.

To drive voltage back to the desired interval, and minimize the aforementioned voltage violation cost $\norm{v(t) - \vnom}_{P_v}^2$ and control cost $\norm{u(t)}_{P_u}^2$, one needs the exact system dynamics \eqref{eq:vol_dyn1} for choosing the optimal reactive power injections $q^c(1), q^c(2), \dotsc$. However, in distribution systems, the exact network parameters are often unknown or hard to estimate. 

In this paper, we work with the voltage control problem on an \emph{unknown} grid topology. We assume that the true $X^\star$ lies within a known compact set $\X \subset \Sym_+^n \cap \R^{n \times n}_+$ and that we only have access to the real-time voltage measurement $v(t)$ at each bus. We perform voltage control while learning the system model at the same time. ($\Sym_+^n$ is the set of $n \times n$ positive semidefinite matrices, and $\R^{n \times n}_+$ is the set of $n \times n$ matrices with nonnegative entries.)

\section{Robust Online Voltage Control}

\begin{figure}[t]
    \centering
    \includegraphics[width=\columnwidth]{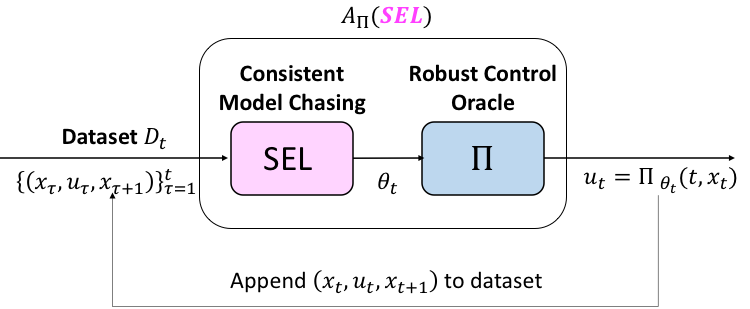}
    \caption{Online robust control framework}
    \label{fig:orc_framework}
\end{figure}

In this section we introduce our robust online voltage control algorithm and its performance bound (\Cref{thm:main}), which is the main result of this paper.

\subsection{Algorithm}

The structure of the algorithm is summarized in \Cref{fig:orc_framework} and detailed in \Cref{alg:robust_online_volt_control}. As the figure shows, the algorithm consists of two main components, a consistent model chasing algorithm \SEL{} and a robust control oracle \CTRL, which are then combined by adapting ideas from \cite{ho2021online}. These two crucial components are detailed in steps (2) and (3) in \Cref{alg:robust_online_volt_control}, respectively. 

The model chasing algorithm \SEL{} performs nested CBC, which is the online problem of choosing a sequence of points within sequentially nested convex sets, with the aim of minimizing the sum of distances between the chosen points \cite{argue2019nearly}. In our setting, the nested convex sets are the consistent sets of possible model parameters, described in \Cref{sec:proof}. We use a simple projection-based algorithm that is more computationally efficient than the Steiner point-based approaches, which have state-of-the-art competitive ratio for nested CBC but are not computationally tractable for high-dimensional settings like voltage control.

The robust control oracle \CTRL{} we use is novel and is developed specifically to provide a provable robustness guarantee (\Cref{thm:oracle}). This robustness guarantee is necessary for the analysis which integrates \SEL{} with \CTRL{} to provide the finite mistake guarantee of the overall algorithm. See \Cref{sec:proof} for details. Note that other choices for both of these components are possible, as long as they provide the guarantees needed in the analysis in \Cref{sec:proof}.

Intuitively, \SEL{} and \CTRL{} are combined in a way such that \CTRL{} outputs an action that causes a voltage limit violation, \SEL{} always reduces the uncertainty about the model by a minimum amount. \SEL{} ensures that our current model estimate $\hat{X}_t$ is consistent with observed data collected so far. This model improvement means that \CTRL{} cannot take too many ``bad'' actions before the system uncertainty is small.

One important detail in \Cref{alg:robust_online_volt_control} is the inclusion of the slack variable $\xi$. If no slack variable is included (equivalently, with $\xi=0$ fixed), the optimization problem in \CTRL{} is guaranteed (\Cref{thm:oracle}) to be feasible and keep the voltage within limits only when the current model estimate $\hat{X}_t$ is close enough to the true model. However, such a guarantee does not necessarily hold when $\xi$ is allowed to be nonzero. On the other hand, since the current model estimate $\hat{X}_t$ may in general be far from the true model, the optimization problem \CTRL{} without a slack variable may not be feasible; that is, there may not be a control action that keeps the predicted voltage within limits under an incorrect model.

Thus, to be precise, for our finite mistake guarantee (\Cref{thm:main}) to hold, the optimization problem for the robust control oracle \CTRL{} should first be solved without the slack variable. This ensures that if $\hat{X}_t$ is sufficiently close enough to the true model, then the algorithm will not make a mistake. In the case that \CTRL{} is infeasible, then it should be solved again with a slack variable, which ensures feasibility. However, solving \CTRL{} twice is unnecessary in practice, and so we have written \Cref{alg:robust_online_volt_control} to reflect its practical implementation.

\begin{algorithm}[t!]
\caption{Online Robust Voltage Controller}
\label{alg:robust_online_volt_control}
Inputs 
\begin{itemize}
    \item fixed squared voltage magnitude at substation: $v^0 \in \R^n$
    \item desired nominal squared voltage magnitude: $\vnom \in \R^n$
    \item limits on the squared voltage magnitude: $\vlims \subset \R^n$
    \item limits on the reactive power injection: $\qlims \subset \R^n$
    \item initial state: $v(1), q^c(0) \in \R^n$
    \item state and action cost matrices: $P_v, P_u \in \Sym^n_+$
    \item compact convex uncertainty set for the model parameter: $\X \subset \Sym_+^n \cap \R^{n \times n}_+$
    \item compact convex uncertainty set for exogenous voltage quantities: $\Vpar \subset \R^n$
    \item upper bound for noise: $\eta > 0$
    \item weight for slack variable: $\beta > 0$
\end{itemize}
Procedure
\begin{enumerate}[leftmargin=*]
    \item Initialize an empty trajectory $D_0 = [\,]$. Set $t = 1$.
    \item If $t = 1$, initialize estimate of model parameters $\hat{X}_1 \in \mathcal{X}$.

    Otherwise, query the model chasing algorithm for a new parameter estimate: $\hat{X}_t \leftarrow \SEL[D_{t-1}]$.
    \begin{subequations}\label{eq:cmc}
    \begin{align}\hspace{-1mm}
        \SEL[D_t]:\ \min_{\hat{X}_t \in \Sym^n} \
        & \| \hat{X}_t - \hat{X}_{t-1} \|_\triangle^2
            \label{eq:cmc_cost} \\
        \text{s.t.} \
        & \hat{X}_t \in \X
            \label{eq:cmc_constr_X} \\
        & \forall (v_i, v_{i+1}, u_i, q^c_i) \in D_t: \nonumber \\
        &\quad -\eta \one \preceq v_{i+1} - v_i - \hat{X}_t u_i \preceq \eta \one
            \label{eq:cmc_constr_consistent} \\
        &\quad v_{i+1} - \hat{X}_t q_i^c \in \Vpar \label{eq:cmc_constr_vpar}
    \end{align}
    \end{subequations}

    \item Query the robust control oracle for the next control action: $u(t) \leftarrow \CTRL_{\hat{X}_t}(v(t))$.
    \begin{subequations}\label{eq:oracle_slack}
    \begin{align}
       \CTRL_{\hat{X}_t}:\ \min_{u \in \R^n,\, \xi \in \R_+}\
       & \norm{\hat{v}' - \vnom}_{P_v}^2 + \norm{u}_{P_u}^2 + \beta \xi ^2
            \label{eq:oracle_slack_cost} \\
        \text{s.t.}\
        & \qmin \preceq q^c(t-1) + u \preceq \qmax
            \label{eq:oracle_slack_constr_u} \\
        & \hat{v}' = v(t) + \hat{X}_t u
            \label{eq:oracle_slack_constr_dynamics} \\
        & k = \eta + \rho(\epsilon) \norm{u}_2
            \label{eq:oracle_slack_constr_buffer} \\
        & \vmin + (k-\xi) \one \preceq \hat{v}' \preceq \vmax - (k-\xi) \one
            \label{eq:oracle_slack_constr_robust}
    \end{align}        
    \end{subequations}

    \item Update $q^c(t) \leftarrow q^c(t-1) + u(t)$. Apply the control action $u(t)$. Observe the system transition to $v(t+1) = v(t) + X^\star u(t) + w(t)$.
    \item Append $(v(t), v(t+1), u(t),q^c(t))$ to the trajectory $D_{t-1}$,
    \[
        D_t = \left[ (v(i), v(i+1), u(i), q^c(i)) \right]_{i=1}^t.
    \]
    \item Increment $t \leftarrow t+1$. Repeat from Step (2).
\end{enumerate}
\end{algorithm}

\subsection{Assumptions}
Before presenting the main results, we introduce several key assumptions that underlie our analysis and discuss why they are both needed and practical. 

The first assumption, stated below, is standard and bounds the noise in the dynamics.

\begin{assumption}
\label{as:bounded_w}
The change in noise is bounded as
\begin{equation}
    \forall t: \norm{w(t)}_\infty \leq \eta,
\label{eq:bounded_w}
\end{equation}
where $w(t) = \vpar(t) - \vpar(t-1)$, $\eta \in [0, \etamax]$ is a known constant and $\etamax = \min_{i=1,\dotsc,n} (\vmax_i - \vmin_i)/2$.
\end{assumption}

Physically, this bound represents an assumption that the active and exogenous reactive power injection does not vary dramatically between time steps, as can be seen by expanding the definition of $w(t)$:
\begin{align*}
    w(t)
    &= \vpar(t) - \vpar(t-1) \\
    &= R^\star(p(t) - p(t-1)) + X^\star(q^e(t) - q^e(t-1)).
\end{align*}
For example, if the net active and exogenous reactive power injection is the same at time steps $t$ and $t-1$, then $w(t) = 0$.

The requirement that $\eta \leq \etamax$ ensures that the space between $\vmin + \eta\one$ and $\vmax - \eta\one$ is nonempty so that, for any $v(t)$ inside this space, the addition of noise $w(t)$ does not push the voltage outside of the target voltage limits. On the other hand, if this space were empty, then for any $v(t)$ and $u(t)$, there would exist some $w(t)$ satisfying $\norm{w(t)}_\infty \leq \eta$ such that $v(t+1)$ exceeds the desired voltage limits. Because we seek a voltage controller that is robust to change in exogenous noise, we must assume that the space is nonempty.

Our second assumption provides a bound on the uncertainty about the network topology, such as the maximum connectivity and impedances.

\begin{assumption}
\label{as:bounded_X}
The true model $X^\star$ lies within a known compact, convex set $\mathcal{X} \subset \Sym_+^n \cap \R_+^{n \times n}$. Let $\diam(\X)$ denote the diameter of $\mathcal{X}$:
\[
    \diam(\X)
    = \sup_{X_1, X_2 \in \mathcal{X}} \norm{X_1 - X_2}_\triangle.
\]
\end{assumption}

\Cref{as:bounded_X} ensures that the unknown true model parameters $X^\star$ belong to a compact, convex set $\X$, which is a minimal assumption necessary for proving an analytic guarantee. This compact set forms the starting point of our consistent model chasing algorithm \SEL.

We equip the space $\Sym^n$ of $n \times n$ symmetric matrices with a norm $\norm{\cdot}_\triangle$ defined as
\[
    \norm{A}_\triangle
    = \norm{\text{upper-triangle}(A)}_2
    = \sqrt{\sum_{i=1}^n \sum_{j=i}^n A_{ij}^2}.
\]
We use this choice of norm in order to isometrically map the space of symmetric matrices to Euclidean space, thereby enabling us to take advantage of known results on nested convex body chasing within Euclidean space.

Finally, our third assumption is about the existence of feasible control actions for the robust control oracle. This assumption can be interpreted as either a bound on the noise, or a requirement that the controllable reactive power injection be flexible enough to satisfy the demand of any noise.

\begin{assumption}
\label{as:bounded_vpar}
There exists a compact, convex set $\Vpar \subset \R^n$ such that $\forall t \geq 0:\ \vpar(t) \in \Vpar$. Furthermore, for some known $\epsilon > 0$,
\begin{align*}
& \forall \vpar \in \Vpar,\, X \in \mathcal{X}. \\
& \exists q^c \in \qlims \text{ s.t. } X q^c + \vpar \in [\vmin + (\eta+\epsilon)\one,\, \vmax - (\eta+\epsilon)\one].
\end{align*}

\end{assumption}

Intuitively, the $\eta$ padding is required for robustness to the noise $w(t)$, while the $\epsilon$ padding is required for robustness to model uncertainty (\emph{i.e.}, uncertainty about $X^\star$).

\subsection{Main result}

We can now state our main result, which is a finite-error bound for \Cref{alg:robust_online_volt_control}.

\begin{theorem}[Main Result]
\label{thm:main}
Under \Cref{as:bounded_w,as:bounded_X,as:bounded_vpar}, \Cref{alg:robust_online_volt_control} ensures that the voltage limits will be violated at most
\begin{equation}\label{eq:error_bound}
   \frac{4 \gamma}{\rho(\epsilon)} \diam(\X) + 1 
\end{equation}
times, where $\rho(\epsilon) = \frac{\epsilon}{2 \norm{\qmax - \qmin}_2}$, $\gamma = \pi (m - 1) m^{m/2}$, and $m = n(n+1)/2$.
\end{theorem}

To the best of our knowledge, this result is the first provable stability bound for voltage control in a setting where the network topology is unknown. It highlights that the \Cref{alg:robust_online_volt_control} can ensure stability even after \emph{unknown} changes to the network topology, \textit{e.g.}, due to maintenance, failures, etc., without the need to perform system identification.

Intuitively, this result guarantees that the model chasing algorithm \SEL{} will learn a ``good enough'' model for control quickly. When the robust controller \CTRL{} makes a mistake, the model chasing algorithm will learn from that mistake and significantly reduce the set of consistent models. Because the initial set of consistent models is bounded, and this set shrinks a significant amount after each mistake, the total number of mistakes is bounded.

To interpret the error bound~\eqref{eq:error_bound} in \Cref{thm:main}, we notice that it is proportional to the diameter $\diam(\X)$ of the parameter space $\X$ and the competitive ratio $\gamma$ of the consistent model chasing algorithm, and inverse proportional to the oracle robustness margin $\rho$. Note that the dependence on $m$ of the consistent model chasing competitive ratio $\gamma$ is very conservative. Because of computational tractability concerns, our implementation of \SEL{} uses a projection-based algorithm rather than the state-of-the-art Steiner point method~\cite{bubeck2020chasing,argue2019nearly}. For the algorithm in~\cite{bubeck2020chasing}, $\gamma_\text{Steiner} = m/2$. As our case studies show, in practice the projection-based algorithm used in \SEL{} performs much better than the worst-case bound.

We outline a proof of \Cref{thm:main} in the next section, but before doing so we want to highlight one piece of that proof that is of independent interest.  In particular, a major step in the proof is to provide a feasibility guarantee for the robust control oracle component \CTRL{} of the algorithm, which is done in \Cref{thm:oracle}.

\section{Proofs}
\label{sec:proof}

We now prove our main result \Cref{thm:main}. Our proof builds on and adapts the approach of~\cite{ho2021online}, which outlines a general framework for integrating model chasing and robust control.  To explain the general framework, we first consider a discrete-time nonlinear dynamical system
\[
    x_{t+1} = f_*(x_t, u_t) + w_t,
    \qquad
    x_0 \text{ given},
    \qquad
    (f_*, \mathbf{w}) \in \mathcal{F},
\]

\noindent where $x \in \mathcal{S} \subseteq \R^n$ is the system state and $u \in \mathcal{U} \subseteq \R^m$ is the control input. The unknown function $f_*$ and disturbance sequence $\mathbf{w} \in \ell^\infty(\Z_+; \R^n)$ belong to an uncertainty set $\mathcal{F}$, and the disturbance is bounded as $\norm{\mathbf{w}}_\infty \leq \eta$. Assume that $\mathcal{F}$ has a \emph{compact parametrization} $(\mathbb{T}, \mathsf{K}, d)$, where $\mathbb{T}: \mathsf{K} \to \wp(\mathcal{F})$ is a mapping from the parameter space $\mathsf{K}$ to a set of functions and disturbances such that $\mathcal{F} \subseteq \bigcup_{\theta \in \mathsf{K}} \mathbb{T}[\theta]$. $\wp(\mathcal{F})$ denotes the powerset of $\mathcal{F}$. Let $d$ denote a metric on $\mathsf{K}$, so $(\mathsf{K}, d)$ is a compact metric space.

The control objective is specified as a sequence of indicator ``goal" functions $\mathcal{G} = (\mathcal{G}_0, \mathcal{G}_1, \dotsc)$. Each $\mathcal{G}_t: \mathcal{X} \times \mathcal{U} \to \{0,1\}$ encodes a desired condition per time step $t$:
\[
    \mathcal{G}_t(x_t, u_t) = \one[\text{$x_t$, $u_t$ violate desired condition at time $t$}].
\]

The main result of \cite{ho2021online} specifies a set of sufficient conditions for a finite-mistake guarantee. These conditions decouple online robust control into separate online learning and robust control components. 

The online learning component requires a consistent model chasing algorithm \SEL, which takes as input the current observed trajectory $D_t = [(x_i, x_{i+1}, u_i)]_{i=1}^t$ and outputs an estimated parameter $\theta_t \in \mathsf{K}$. The estimated parameter $\theta_t$ must be \emph{consistent} with $D_t$. \begin{definition}[Consistent Parameter]
\label{def:consistent}
We say $\theta \in \mathsf{K}$ is consistent with $D_t$ if there exists 
$(f,\mathbf{w}) \in \mathbb{T}[\theta]$ such that
\[
    \forall (x_t, x_{t+1}, u_t) \in D_t:\ x_{t+1} = f(x_t,u_t; \theta) + w_t.
\]
\end{definition}
Let $P_t$ denote the set of all parameters consistent with $D_t$. We say \SEL{} is $\gamma$-competitive if $\sum_{t=1}^{\infty} d(\theta_t, \theta_{t-1}) \leq \gamma \max_{\theta \in \mathsf{K}} d(P_{\infty}, \theta)$ holds for a fixed constant $\gamma>0$, which we call the \emph{competitive ratio}.

The robust control component requires a control oracle \CTRL, which given the current state $x_t$ and a parameter $\theta_t$, outputs a control action $u_t = \CTRL_{\theta_t}(x_t)$ that is robust for all systems that are close to $\theta_t$. In particular, we call a control oracle \emph{$\rho$-robust} for control objective $\mathcal{G}$, if all trajectories in $S^\Pi[\rho; \theta]$ achieve $\mathcal{G}$ after finitely many mistakes, and $S^\Pi[\rho; \theta]$ is defined as the set of all possible trajectories generated by $\CTRL_{\hat\theta}$ for all $\hat\theta$ such that $d(\theta, \hat\theta) \leq \rho$.
\begin{equation}     S^\Pi[\rho; \theta] = \left\{
           \begin{aligned}
            & D_\infty = [(x_t, x_{t+1}, u_t)]_{t=1}^\infty: \\
            &\quad u_t = \Pi_{\hat{\theta}}(x_t) \\
            &\quad x_{t+1} = f(x_t, u_t) + w_t
        \end{aligned}
        \ \middle| \
        \begin{aligned}
        & (f, \mathbf{w}) \in \mathbb{T}[\theta], \\
        & d(\hat{\theta}, \theta) \leq \rho
        \end{aligned}
    \right\}
\end{equation}
Due to the page limit, we refer readers to~\cite{ho2021online} for a more detailed discussion of consistent model chasing algorithms and $\rho$-robust control oracles. As a summary, if \SEL{} chases consistent models and \CTRL \  is a robust oracle for $\mathcal{G}$, then the resulting $A_\CTRL(\SEL)$ algorithm achieves a finite mistake guarantee, which is stated in the following.

\begin{theorem} \label{thm:ho}
\cite[Theorem 2.5]{ho2021online} Assume that \SEL{} chases consistent models and \CTRL{} is a robust oracle for objective $\mathcal{G}$. Then for any starting point $x_0$ and trajectory $[(x_t, u_t)]_{t=0}^\infty$ generated by $\mathcal{A}_\CTRL(\SEL)$ (illustrated in \Cref{fig:orc_framework}), the following mistake guarantees hold: (i) If \CTRL{} is robust, then $\sum_{t=0}^\infty \mathcal{G}_t(x_t, u_t) < \infty$; (ii) If \CTRL{} is uniformly $\rho$-robust and \SEL{} is $\gamma$-competitive, then $$\sum_{t=0}^\infty \mathcal{G}_t(x_t, u_t) < \max\left\{1, M_{\rho}^\CTRL\right\}\left(\frac{2\gamma}{\rho} \diam(\mathsf{K}) +1\right)$$ where $M_{\rho}^\CTRL$ denotes the worst case total mistakes of the $\rho$-robust control oracle \CTRL.
\end{theorem}

To apply \Cref{thm:ho} to prove \Cref{thm:main} we need to prove that (i) the proposed algorithm \Cref{eq:cmc} chases consistent models and has a bounded competitive ratio, and (ii) the proposed robust algorithm in \Cref{eq:oracle} is a $\rho$-robust control oracle, for bounded disturbance in the system topology. 
In particular, the correspondence of the definitions are as follows. We have $\theta = X$, and 
\begin{align*}
    \mathsf{K} &= \X, \quad v(1), q^c(0) \text{ given} \\
    d(X, X') &= \|X - X'\|_\triangle \\
    \mathbb{T}[X] &= \set{
        (f, \mathbf{w}) \mid
        f(v, u) = v + X u,\
        \norm{\mathbf{w}}_\infty \leq \eta
    } \\
    \mathcal{F} &= \bigcup_{X \in \X} T[X] \\
    \mathcal{G}_t &= \one[v(t) \in \vlims].
\end{align*}

We begin by proving that the set defined from \eqref{eq:cmc_constr_X}, \eqref{eq:cmc_constr_consistent}, \eqref{eq:cmc_constr_vpar} in \Cref{alg:robust_online_volt_control} is consistent with the trajectory $D_T$.

\begin{lemma}[$\SEL{}$ is consistent]
Suppose $D_T$ is a trajectory of voltage measurements and control actions taken up to time $T$:
\[
    D_T = [v(t), v(t+1), u(t), q^c(t)]_{t=1}^T.
\]
The set 
\begin{equation}
\label{eq:consistent-set}
    P_T := \set{
        \hat{X} \in \X
        \ \middle| \
        \begin{aligned}
        & \forall\left[v(t), v({t+1}), u(t), q^c(t)\right] \in D_T: \\
        &\quad \norm{v({t+1}) - v(t) - \hat{X} u(t)}_\infty \leq \eta, \\
        &\quad v({t+1}) - \hat{X} q^c(t) \in \Vpar
        \end{aligned}
    }
\end{equation}
is a consistent set for $D_T$, i.e., $\hat{X}\in\mathcal{X}$ is consistent (\Cref{def:consistent}) if and only if $\hat{X}\in P_T$.

\end{lemma}
\begin{proof}
Consider any $\hat{X} \in P_T$. For $t \in \{1, \dotsc, T\}$, define $$\hat{w}(t) := v({t+1}) - v(t) - \hat{X} u(t).$$  Then, $\hat{w}(t)$ clearly satisfies \Cref{as:bounded_w}. Moreover, let $$\vpar(t) := v({t+1}) - \hat{X} q^c(t).$$ Then $\vpar(t)$ is the corresponding admissible noise that satisfies \Cref{as:bounded_vpar}. 

Conversely, if $\hat{X} \not\in P_T$, then either $\norm{v({t+1}) - v(t) - \hat{X} u(t)}_\infty > \eta$ for some $t$ (and therefore $\norm{\hat{w}(t)}_\infty > \eta$) or $\vpar(t)$ violates \Cref{as:bounded_vpar}. Thus, $P_T$ contains exactly all $X \in \mathcal{X}$ that are consistent with $D_T$.
\end{proof}

Observe that $P_T$ is a closed, bounded, and convex set. Furthermore, it is non-empty, since $X^\star \in P_T$. 
Intuitively, $P_T$ is the smallest set containing all parameters that could generate the observed trajectory $\{v(t)\}_{t=1}^{T+1}$ along with a corresponding admissible sequence of noise compatible with \Cref{as:bounded_w,as:bounded_X,as:bounded_vpar}.

Now that we have defined the consistent sets $P_t$, we can express \SEL{} equivalently as solving $\min_{\hat{X}_t \in \Sym^n} \norm{\hat{X}_t - \hat{X}_{t-1}}_\triangle^2$ s.t. $\hat{X}_t \in P_t$. This is a nested convex body chasing problem, where we aim to minimize the movement distance $\norm{\hat{X}_t - \hat{X}_{t-1}}_\triangle$ between nested convex sets $P_t \subseteq P_{t-1}$. By leveraging known results about nested convex body chasing algorithms \cite{argue2019nearly,bubeck2020chasing}, we can prove that the model chasing algorithm \SEL{} described by \eqref{eq:cmc} has a bounded competitive ratio. This is formalized in the following lemma.
\begin{lemma}[\SEL{} is competitive]
\label{thm:ncbc}
For any compact convex space $\mathsf{K} \subset \Sym^n$, the greedy projection algorithm for consistent model chasing (CMC) in \Cref{eq:cmc} achieves a competitive ratio
\[
\gamma = \pi (m-1) m^{m/2}
\]
where $m = n(n+1)/2$.
\end{lemma}

\begin{proof} The normed vector space $(\Sym^n, \norm{\cdot}_\triangle)$ is isometrically isomorphic to the Euclidean space $(\R^m, \norm{\cdot}_2)$ with $m = n(n+1)/2$. The mapping between the two spaces is the vectorization of the upper-triangle of the symmetric matrix:
\[
    A \in \Sym^n \longleftrightarrow [A_{1,1:n}, A_{2,2:n}, \dotsc, A_{n,n}]^\top \in \R^m.
\]
Note that $\forall d \in \N:\ (\omega_d / \omega_{d-1}) \leq \pi$, where $\omega_d$ denotes the surface area of the $d$-sphere.
Then, by \cite[Theorem 1.3]{argue2019nearly}, the greedy projection algorithm achieves a competitive ratio of at most $\pi (m-1)m^{m/2}$.
\end{proof}

Finally, we show that our controller \CTRL{} is $\rho$-robust. In particular, we prove that $\CTRL_{\hat{X}}$ makes no mistakes ($M_\rho^\Pi = 0$) given a consistent model $\hat{X} \in P_t$.

\begin{theorem}[\CTRL{} is $\rho$-robust]
\label{thm:oracle}
Under \Cref{as:bounded_w,as:bounded_X,as:bounded_vpar}, suppose $\hat{X} \in P_t$, where $P_t$ is given in \eqref{eq:consistent-set} for $t\geq 1$. Define $\rho(\epsilon) = \epsilon/\left(2\|\qmax - \qmin\|_2\right)$. 
Then, the following optimization problem is feasible,
\begin{subequations}
\label{eq:oracle}
\begin{align}
    \min_{u \in \R^n}\quad & \norm{\hat{v}' - \vnom}_{P_v}^2 + \norm{u}_{P_u}^2 \label{eq:oracle_cost} \\
    \text{s.t.}\quad
    & \qmin \preceq q^c(t-1) + u \preceq \qmax
        \label{eq:oracle_constr_u} \\
    & \hat{v}' = v(t) + \hat{X}u
        \label{eq:oracle_constr_dynamics} \\
    & k = \eta + \rho(\epsilon) \norm{u}_2
        \label{eq:oracle_constr_buffer} \\
    & \vmin + k \one \preceq \hat{v}' \preceq \vmax - k \one.
        \label{eq:oracle_constr_robust}
\end{align}
\end{subequations}
Further, the solution of \eqref{eq:oracle}, $u(t)$, guarantees voltage stability for all $X \in \X$ such that $\norm{X-\hat{X}}_\triangle \leq \rho(\epsilon)$. That is, $v(t) + X u(t) +w(t)\in \vlims$ for all $w(t)$ satisfying \Cref{as:bounded_w}.
\end{theorem}

Observe that \Cref{eq:oracle} corresponds to \Cref{eq:oracle_slack} in \Cref{alg:robust_online_volt_control} with the slack variable set to zero. 
We note that the robustness margin $\rho$ decreases as $\qlims$ increase. The intuitive reason is that the voltage is more sensitive to changes in $\hat{X}$ when the range of possible $u$'s expands. Therefore, a fixed voltage buffer of $\epsilon$ in constraints \eqref{eq:oracle_slack_constr_buffer} and \eqref{eq:oracle_constr_robust} affords less robustness to changes in $\hat{X}$ as $\qlims$ gets larger.

\begin{proof}[Proof of \Cref{thm:oracle}]
First, we prove that the optimization problem \eqref{eq:oracle} is feasible. Define a new variable $q^c := q^c(t-1) + u$. Let
\begin{align}
    \vpar(t-1)
    &:= v(t) - \hat{X} q^c(t-1)
    \\
    \hat{v}'(u)
    &:= v(t) + \hat{X} u
    = v(t) + \hat{X}[q^c - q^c(t-1)] \\
        &= \hat{X} q^c + \vpar(t-1)
    \\
    {v}'(u)
    &:= v(t) + X u
    = \hat{v}'(u) + (X - \hat{X}) u. \label{eq:v_hat_prime}
\end{align}
Here, $\vpar$ is the conjectured admissible noise when we assume the underlying parameter is $\hat{X}$. Recall \eqref{eq:vol_dyn2} and we can interpret $\hat{v}'$ as the one-step prediction of voltage under the selected consistent model $\hat{X}$ given a control action $u$ and the current voltage $v(t)$, without the disturbance term. Similarly, $v'$ is the disturbance-free one-step prediction of voltage under a different model $X$ such that $\norm{X-\hat{X}}_\triangle \leq \rho(\epsilon)$ given the same control action $u$ and the current voltage $v(t)$.

Since $\hat{X}$ is consistent with $P_{t-1}$, $\vpar(t-1) \in \Vpar$.
Therefore, by \Cref{as:bounded_vpar}, there exists $q^c \in \qlims$ such that
\begin{equation}
    \label{eq:as3_bound_in_thm}
    \vmin + (\eta+\epsilon)\one
    \preceq \hat{v}'(u)
    \preceq \vmax - (\eta+\epsilon)\one.
\end{equation}
Then the $u$ corresponding to such $q^c$, \emph{i.e.}, $u = q^c - q^c(t-1)$, clearly satisfies constraint \eqref{eq:oracle_constr_u}. By \Cref{lemma:tri_norm}, for all ${X}$ satisfying $\|{X} - \hat{X}\|_\triangle \leq \rho(\epsilon)$,
\begin{equation}\label{eq:xhat_x}
    -\rho(\epsilon) \norm{u}_2 \one \preceq ({X} - \hat{X}) u \preceq \rho(\epsilon) \norm{u}_2 \one.
\end{equation}
Adding \eqref{eq:xhat_x} with \eqref{eq:as3_bound_in_thm} and using \eqref{eq:v_hat_prime} yields
\begin{equation}
\label{eq:vv}
    \vmin + (\eta + \epsilon - \rho(\epsilon) \norm{u}_2) \one
    \preceq v'(u)
    \preceq \vmax - (\eta + \epsilon - \rho(\epsilon) \norm{u}_2) \one.
\end{equation}
By choosing
\[
    \rho(\epsilon) = \frac{\epsilon}{2 \norm{\qmax - \qmin}_2},
\]
and observing that $u$ must satisfy $\norm{u}_2 \leq \norm{\qmax - \qmin}_2$ since $u = q^c - q^c(t-1)$, we have
\[
    \rho(\epsilon) \norm{u}_2
    \leq \frac{\epsilon}{2}
    \leq \epsilon - \rho(\epsilon) \norm{u}_2.
\]
Therefore, we can use this relation to upper and lower bound \eqref{eq:vv} and arrive at
\begin{equation}
\label{eq:vv-final}
    \vmin + (\eta + \rho(\epsilon) \norm{u}_2) \one
    \preceq {v}'(u)
    \preceq \vmax - (\eta + \rho(\epsilon) \norm{u}_2) \one
\end{equation}
This holds for any one-step noiseless prediction $v'(u)$ from $X$ such that $\norm{X-\hat{X}}_\triangle \leq \rho(\epsilon)$. Since $\hat{X}$ trivially lies in this set, $\hat{v}'$ also satisfies \eqref{eq:vv-final}, which in turn shows that $u$ satisfies constraint \eqref{eq:oracle_constr_robust}. Thus, the optimization problem is feasible.

Next, we show that every solution $u$ from \eqref{eq:oracle} generated with $\hat{X} $ guarantees that $v(t) + X u \in \vlims$ for all $X$ such that $\norm{X-\hat{X}}_\triangle \leq \rho(\epsilon)$. Subtracting \eqref{eq:xhat_x} from constraint \eqref{eq:oracle_constr_robust} and using \Cref{eq:v_hat_prime} yields
\[
    \vmin + \eta \one
    \preceq v'(u)
    \preceq \vmax - \eta \one.
\]
By \Cref{as:bounded_w}, any admissible disturbance $w(t)$ is bounded as $-\eta \one \preceq w(t) \preceq \eta \one$. This means that 
\[
    \vmin
    \preceq v(t+1) = v'(u) + w(t)
    \preceq \vmax,
\]
which shows the control action computed using \eqref{eq:oracle} guarantees voltage stability for all $X \in \mathcal{X}$ such that $\norm{X-\hat{X}}_\triangle \leq \rho(\epsilon)$.
\end{proof}

\begin{lemma}\label{lemma:tri_norm}
For all $A \in \Sym^n$, $b \in \R^n$, and $\alpha \in \R_+$,
\[
    \norm{A}_\triangle \leq \alpha
    \quad\implies\quad
    -\alpha \norm{b}_2 \one \preceq A b \preceq \alpha \norm{b}_2 \one.
\]
\end{lemma}
\begin{proof}
Let $A_i$ denote the $i^\text{th}$ row of $A$. By symmetry of $A$,
\begin{align*}
    \norm{A_i}_2^2
    &= \sum_{j=1}^n A_{i,j}^2
    = \sum_{k=1}^{i-1} A_{k,i}^2 + \sum_{j=i}^n A_{i,j}^2 \\
    &\leq \sum_{k=1}^n \sum_{j=k}^n A_{k,j}^2
    = \norm{A}_\triangle^2
    \leq \alpha^2,
\end{align*}
so $\norm{A_i}_2 \leq \alpha$. Then
\[
    -\alpha \norm{b}_2
    \leq -\norm{A_i}_2 \norm{b}_2
    \leq (Ab)_i
    \leq \norm{A_i}_2 \norm{b}_2
    \leq \alpha \norm{b}_2.
\]
This holds for all $i \in \{1, \dotsc, n\}$, which yields the desired result.
\end{proof}

Finally, combining \Cref{thm:oracle} with \Cref{thm:ncbc} and applying \Cref{thm:ho} completes the proof of \Cref{thm:main}.
\section{Case Study}
\label{sec:experiments}
We demonstrate the effectiveness of \Cref{alg:robust_online_volt_control} using a case study based on a single-phase 56-bus network from the Southern California Edison (SCE) utility (\Cref{fig:sce56bus}). We use a setup that mimics what has been used previously in the voltage control literature.  In particular, the detailed line parameters $r_{ij}$ and $x_{ij}$ that we use for this network are the same as those in Table 1 of~\cite{farivar2012optimal}. In our experiments, we use the linear power model in \Cref{eq:simplified_distflow} to solve for voltages; we do not use non-linear power flow models.

\begin{figure}[tbp]
    \centering
    \includegraphics[width=\columnwidth]{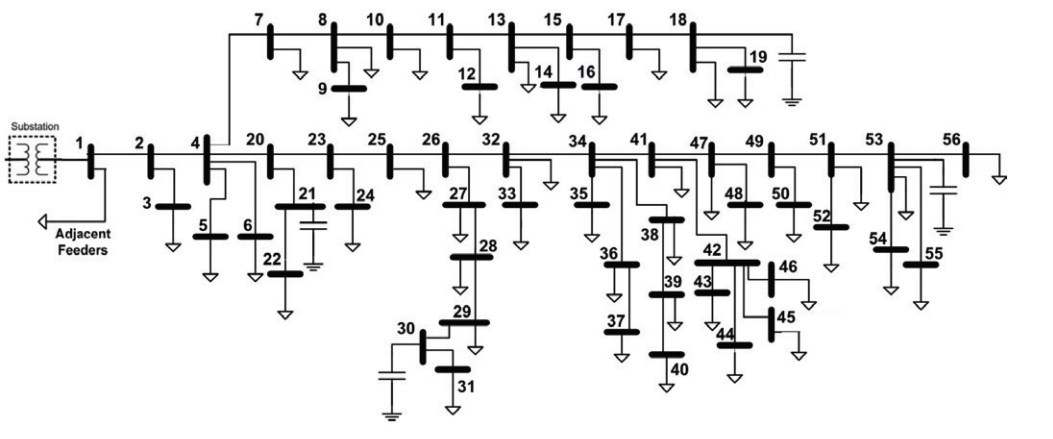}
    \caption{Schematic diagram of SCE 56 bus distribution system.}
    \label{fig:sce56bus}
\end{figure}

\subsection{Experimental Setup}
We use real-world data~\cite{qu2020optimal} collected from the unmodified network, with power injection at buses in $\mathcal{C} = \{2, 4, 7, 8, 9, 10, 11, 12, 13, 14, 15,\\ 16, 19, 20, 23, 25, 26, 32\}$ adjusted by scaled actual generation from a nearby photovoltaic system. Exogenous active and reactive power injection measurements are taken at each bus at 6-second intervals over a 24-hour period. We assume that controllers with reactive power injection capacity are set up at every node. \Cref{fig:baseline_noctrl} plots these values for several buses to illustrate the setting considered.

The network parameters used in our experiments are as follows:
\begin{itemize}
    \item nominal squared voltage magnitude at the substation $v^0 = \vnom = (12\si{kV})^2$
    \item squared voltage magnitude limits $\vlims = [0.95 \text{pu}, 1.05 \text{pu}] = [(11.4 \si{kV})^2, (12.6 \si{kV})^2]$
    \item reactive power injection limits $\qlims = [-0.24 \text{MVar}, 0.24 \text{MVar}]$
    \item state and input cost matrices $P_v = 0.1 I$, $P_u = 10 I$
    \item initial state $v(1) = Rp(0) + X q^e(0) + v^0 \one$, $q^c(0) = \zero$
\end{itemize}

Note that, in comparison to previous papers in the voltage control literature, our reactive power injection limits $\qlims$ are slightly more generous than the $\pm 0.2$ MVar used in, e.g., \cite{qu2020optimal}. We choose $\pm 0.24$ MVar because even a controller with perfect knowledge of the future would need reactive power injection capabilities of at least $\pm 0.238$ MVar in order to maintain $v(t) \in [\vmin, \vmax]$ (if $\qmin = -\qmax$).

In order to truly satisfy the requirement in \Cref{as:bounded_vpar} $v(t) \in [\vmin + (\eta+\epsilon), \vmax - (\eta+\epsilon)]$ with $\epsilon=0.1$, the reactive power injection capabilities need to be at least $\pm 0.455$ MVar. As we show in our experiments with only $\pm 0.24$ MVar range of control, \Cref{as:bounded_vpar} does not need to be fully satisfied in order for our method to still provide strong empirical results.

We fix $\epsilon=0.1$ in our experiments, corresponding to a robustness margin $\rho = \epsilon / (2 \|\qmax-\qmin\|_2) \approx 0.014$. We set $\eta = 8.65$, which upper-bounds the true maximum change in exogenous noise observed in our data:
\[
    \max_t \norm{R (p(t) - p(t-1)) + X (q^e(t) - q^e(t-1))}_\infty
    \approx 8.625.
\]
For the robust controller \CTRL, we use a slack variable weight $\beta = 100$ in the cost function, and we set $\Vpar = [\underline\vpar, \overline\vpar]$ to be the rectangle around the true noise:
\[
    \forall i \in \{1,\dotsc,n=55\}:\quad
    \underline\vpar_i = \min_t \vpar_i(t),\quad
    \overline\vpar_i = \max_t \vpar_i(t).
\]

In our implementation of the consistent model chasing algorithm, we make a few changes from the procedure described in \Cref{alg:robust_online_volt_control}. In order to keep the consistent model chasing optimization problem \eqref{eq:cmc} computationally tractable, we do not use the full trajectory $D$ as in constraints \eqref{eq:cmc_constr_consistent}-\eqref{eq:cmc_constr_vpar}. Instead, we include the 20 latest observations and a set of 80 more random samples $(v(t), v(t+1), u(t), q^c(t)) \sim D$. This provides a computationally tractable approximation of the uncertainty set.

Our experiments aim to understand the performance of the proposed online robust controller under different levels of uncertainty and the effect of different initializations of $\hat{X}$. In particular, we consider uncertainty sets of the form
\begin{equation}
    \X_\alpha
    = \set{
        \hat{X} \in \Sym_+^n \cap \R_+^{n \times n}:\
        \|\hat{X} - X^\star\|_\triangle \leq \alpha \norm{X^\star}_\triangle
    }
\end{equation}
where larger $\alpha$ corresponds to a larger uncertainty set. Note that $\diam(\X_\alpha) = 2 \alpha \norm{X^\star}_\triangle$.

We initialize $\hat{X}$ by adding noise to the true $X^\star$ in two ways. First, we multiply the line impedance coefficients $x_{ij}$ by a random scaling factor sampled from $\textsf{Uniform}[1-\sigma, 1+\sigma]$ for some $\sigma \in [0,1]$. Second, we randomly permute the bus ordering, so that $\hat{X}$ corresponds to a permuted grid topology. Finally, we project $\hat{X}$ into the uncertainty set $\X_\alpha$.

We consider 3 settings:
\begin{enumerate}
    \item moderate uncertainty and error ($\alpha = 0.5$), with $\sigma = 0.5$
    \item large uncertainty and error ($\alpha = 1.0$), with $\sigma = 1.0$
    \item large uncertainty ($\alpha = 1.0$) with $\sigma = 1.0$, but with moderate error (the initial $\hat{X}$ is projected into $\X_{0.5}$)
\end{enumerate}

By comparing case (3) to cases (1) and (2), we distinguish the impact of the uncertainty set size and the error of the initial guess for $\hat{X}$, each of which presents different challenges for \Cref{alg:robust_online_volt_control}.

We compare our experiments against four baselines, illustrated in \Cref{fig:baseline}: (a) the case of no controller, (b) the case where our robust controller is used with perfect knowledge of the network topology, (c) a model-agnostic decentralized controller from \cite[Section IV]{li2014realtime}, and (d) another model-agnostic decentralized controller from \cite[Section V]{li2014realtime}. The figure highlights that, without a controller, buses 19 and 31 violate the upper and lower voltage limits, respectively, by a significant margin. In contrast, the robust controller given the true $X^\star$ keeps the voltage within the limits for all buses, as expected. The decentralized model-agnostic controllers (c) and (d), which are supposed to be robust to the underlying topology, do not perform well, as their theoretical convergence guarantees \cite{li2014realtime} only hold for fixed $\vpar$ and for unbounded reactive power injection limits $(\qmin, \qmax)$.

\begin{figure*}[tbp]
 \centering
 \begin{subfigure}[b]{0.23\textwidth}
     \centering
     \includegraphics[width=\textwidth]{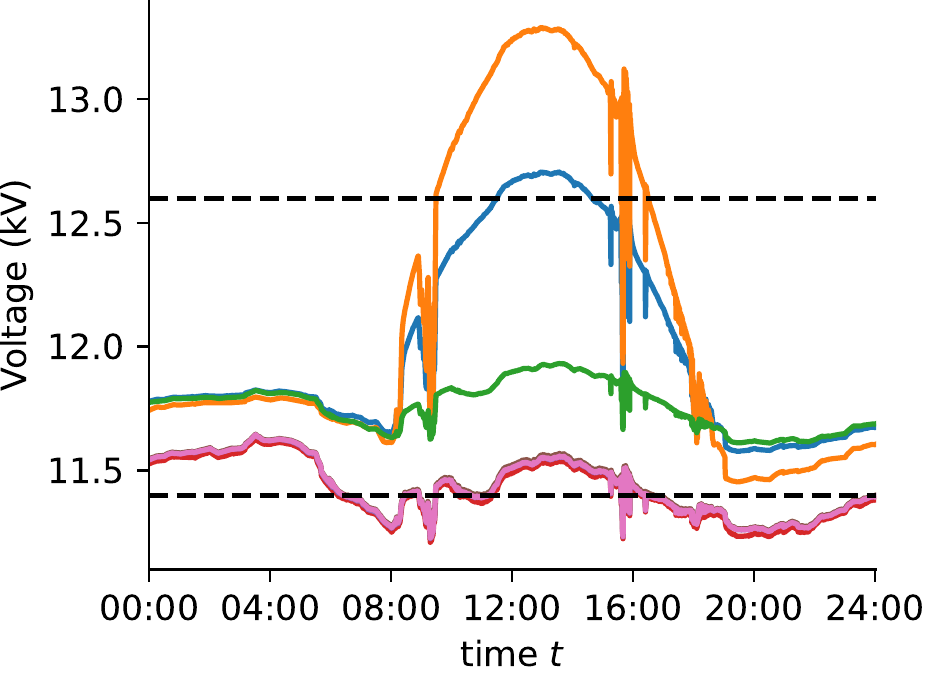}
     \caption{}
     \label{fig:baseline_noctrl}
 \end{subfigure}
 \hfill
 \begin{subfigure}[b]{0.23\textwidth}
     \centering
     \includegraphics[width=\textwidth]{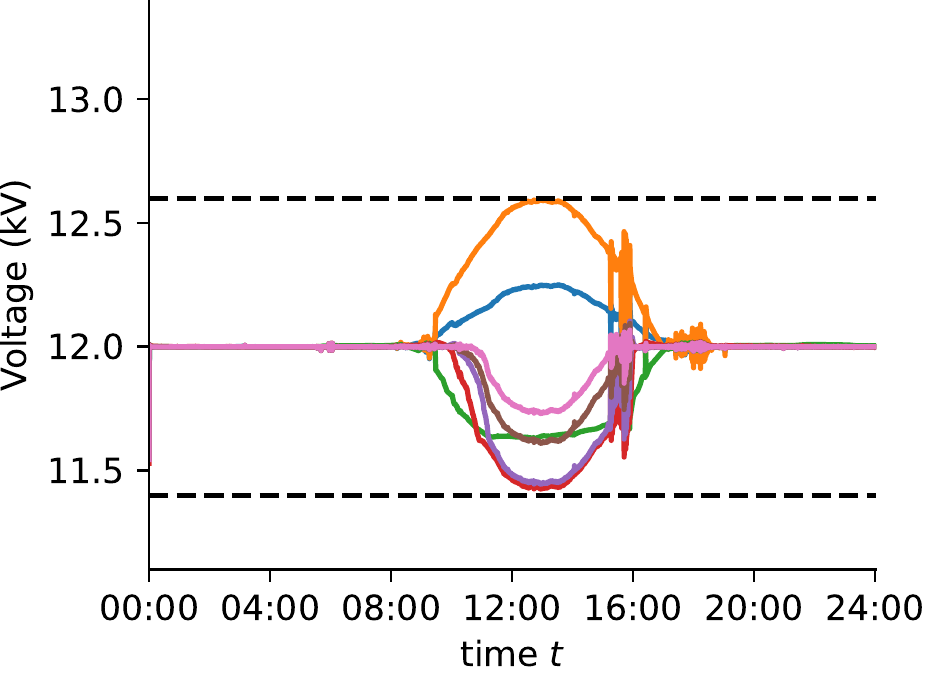}
     \caption{}
     \label{fig:baseline_perfect}
 \end{subfigure}
 \hfill
 \begin{subfigure}[b]{0.23\textwidth}
     \centering
     \includegraphics[width=\textwidth]{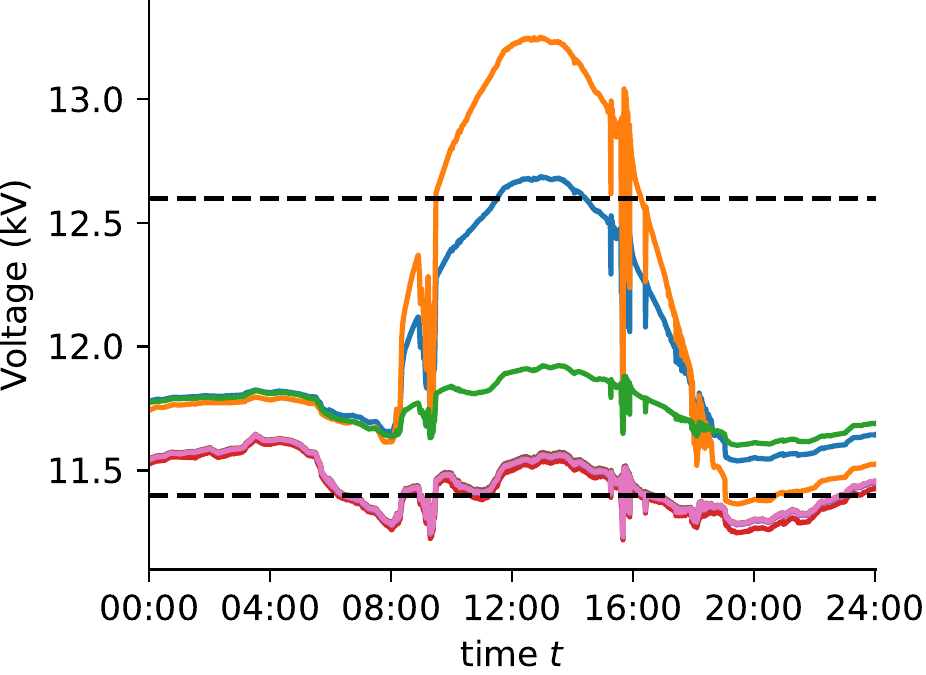}
     \caption{}
     \label{fig:baseline_li_feas}
 \end{subfigure}
 \hfill
  \begin{subfigure}[b]{0.23\textwidth}
     \centering
     \includegraphics[width=\textwidth]{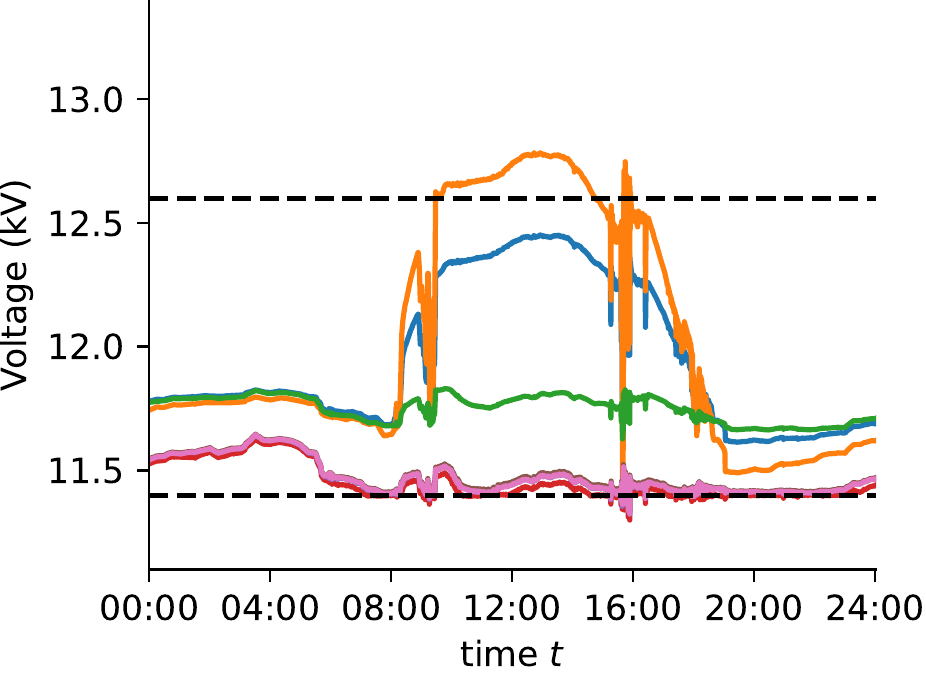}
     \caption{}
     \label{fig:baseline_li_opt}
 \end{subfigure}
  \hfill
  \begin{subfigure}[b]{0.06\textwidth}
     \centering
     \includegraphics[width=\textwidth]{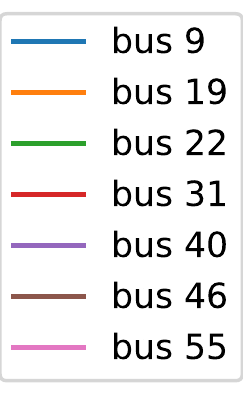}
     \vspace{3em}
 \end{subfigure}
    \caption{Voltage profile of SCE 56 bus distribution system with PV generators. (a) without control (b) robust controller \CTRL{} given the true $X^\star$ (c) model-agnostic controller from \cite[Section IV]{li2014realtime} (d) model-agnostic controller from \cite[Section V]{li2014realtime}.}
    \label{fig:baseline}
\end{figure*}

\subsection{Experimental Results}

Our experimental results focus on demonstrating the ability of \Cref{alg:robust_online_volt_control} to stabilize the system without knowledge of the network topology.  To highlight the performance of the algorithm, we consider settings with moderate and large amounts of uncertainty in \Cref{fig:med_uncertainty,fig:large_uncertainty}.  Importantly, \Cref{alg:robust_online_volt_control} stabilizes the system without performing system identification.  In fact, our results highlight that the algorithm still has significant uncertainty about the topology at the end of the experiments, while also providing nearly the same stabilization performance as the robust controller does with complete information about the topology.

Our simulations use the linearized system dynamics \Cref{eq:vol_dyn1}, and the convex optimization problems for \SEL{} and \CTRL{} are solved with CVXPY \cite{diamond2016cvxpy} using the MOSEK solver. Code for our simulations will be made available on GitHub.

\begin{figure*}[tbp]
 \centering
 \begin{subfigure}[b]{0.33\textwidth}
     \centering
     \includegraphics[width=\textwidth]{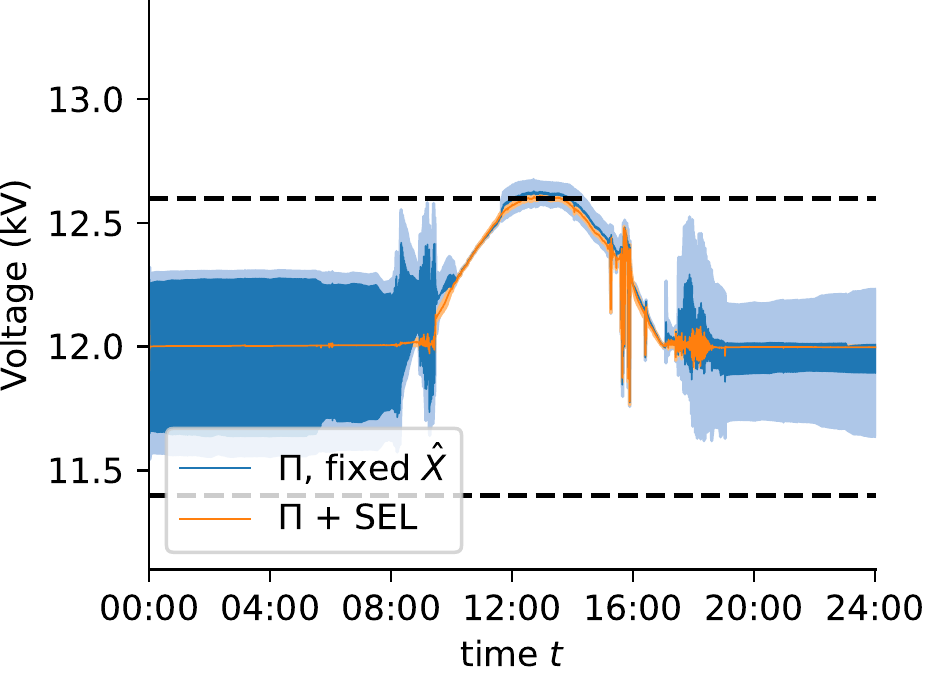}
     \caption{}
     \label{fig:med_uncertainty_bus19}
 \end{subfigure}
 \hfill
 \begin{subfigure}[b]{0.33\textwidth}
     \centering
     \includegraphics[width=\textwidth]{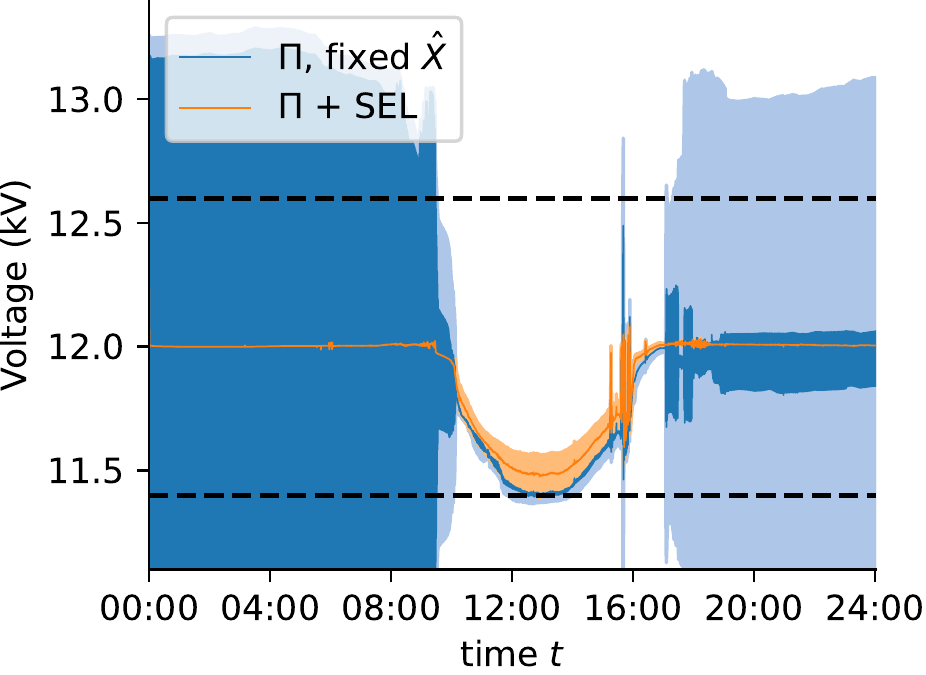}
     \caption{}
     \label{fig:med_uncertainty_bus31}
 \end{subfigure}
 \hfill
 \begin{subfigure}[b]{0.33\textwidth}
     \centering
     \includegraphics[width=\textwidth]{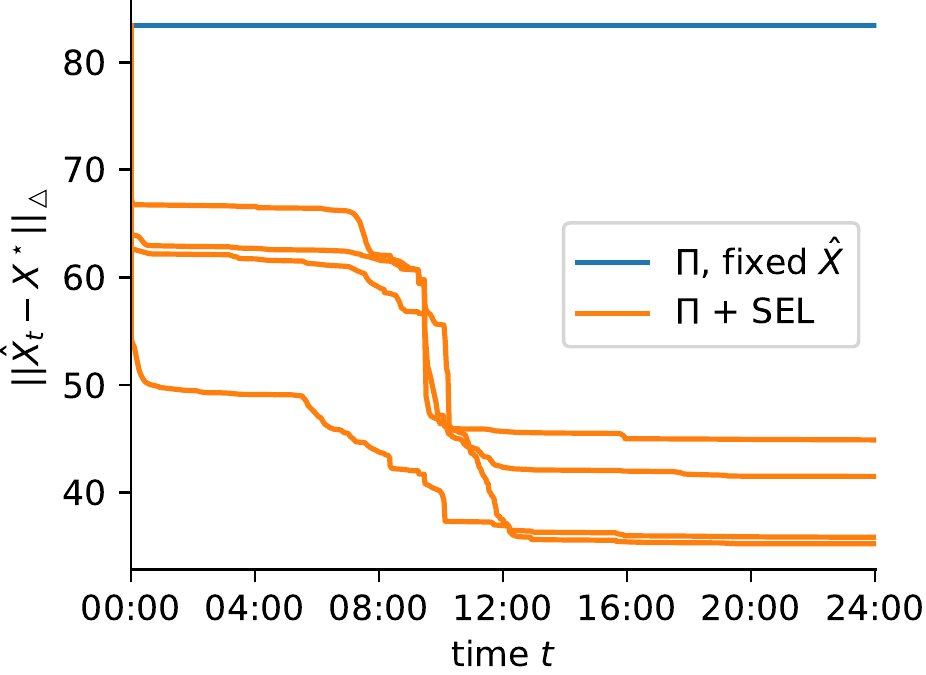}
     \caption{}
     \label{fig:med_uncertainty_cmc}
 \end{subfigure}
    \caption{Results for \Cref{alg:robust_online_volt_control} on the SCE 56-bus distribution system with moderate-sized uncertainty set ($\alpha = 0.5$) and moderate initial error ($\|\hat{X}-X^\star\|_\triangle = 0.5 \norm{X^\star}_\triangle$). (a) Voltage profile of bus 19, where dark blue and dark orange lines plot the mean voltage for the robust controller \CTRL{} with fixed $\hat{X}$ vs. \CTRL{} paired with \SEL, respectively, across 4 random initializations of $\hat{X}$; shaded light blue and light orange regions indicate $\pm 1$ standard deviation. The dark blue line looks like a blob because \CTRL{} over-corrects its mistakes when given incorrect network parameters. (b) Same as (a), but for bus 31. (c) Model uncertainty decreases as \SEL{} learns over time, for the same 4 choices of $\hat{X}$ from (a) and (b). The algorithm keeps the voltage within limits even though the model estimates $\hat{X}_t$ are imperfect, demonstrating that complete system identification is not necessary.}
    \label{fig:med_uncertainty}
\end{figure*}

\begin{figure*}[tbp]
 \centering
 \begin{subfigure}[b]{0.33\textwidth}
     \centering
     \includegraphics[width=\textwidth]{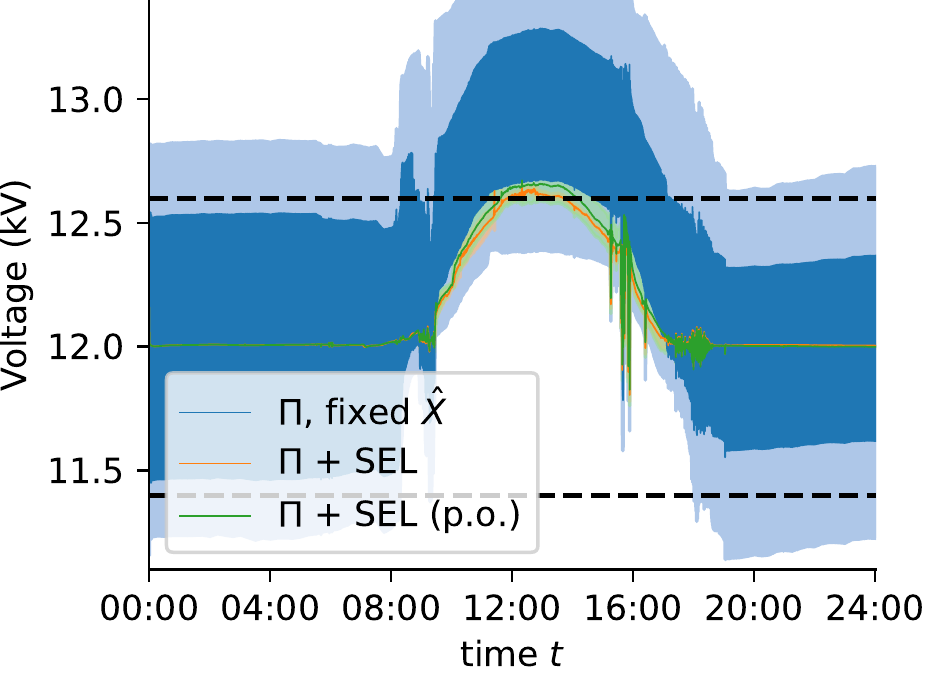}
     \caption{}
     \label{fig:large_uncertainty_bus19}
 \end{subfigure}
 \hfill
 \begin{subfigure}[b]{0.33\textwidth}
     \centering
     \includegraphics[width=\textwidth]{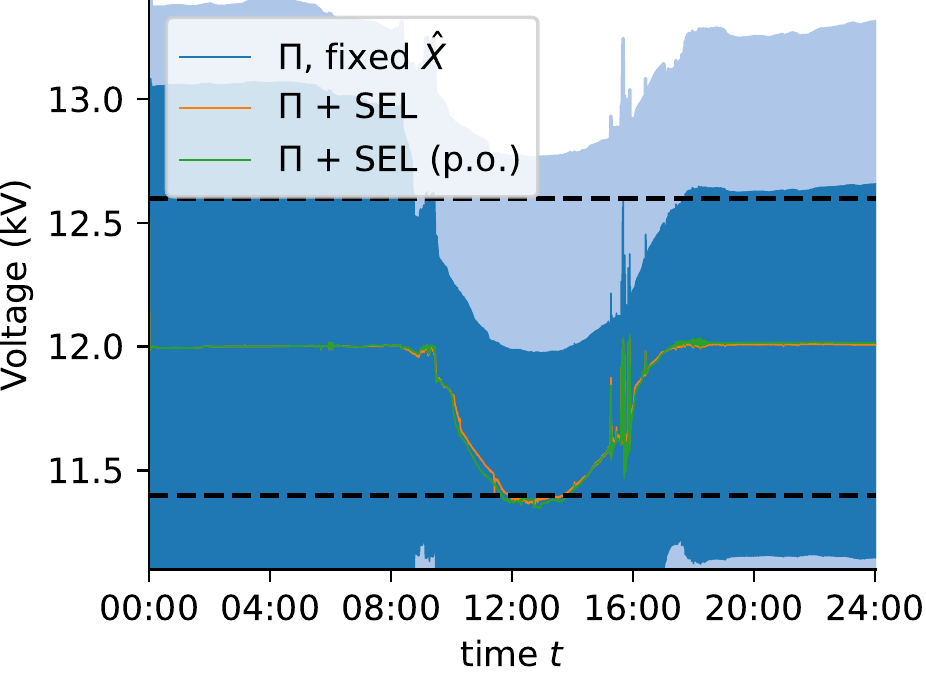}
     \caption{}
     \label{fig:large_uncertainty_bus31}
 \end{subfigure}
 \hfill
 \begin{subfigure}[b]{0.33\textwidth}
     \centering
     \includegraphics[width=\textwidth]{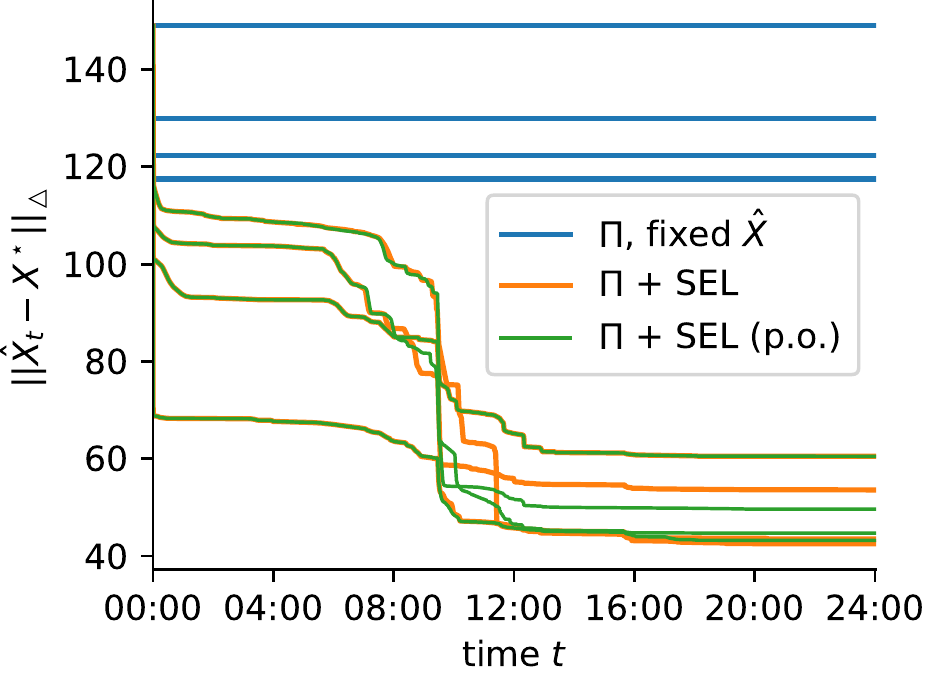}
     \caption{}
     \label{fig:large_uncertainty_cmc}
 \end{subfigure}
    \caption{Same as \Cref{fig:med_uncertainty}, except the uncertainty set and initial error are both larger ($\alpha = 1.0$, $\|\hat{X}-X^\star\|_\triangle = \norm{X^\star}_\triangle$.) The robust controller \CTRL{} with fixed $\hat{X}$ performs significantly worse than in \Cref{fig:med_uncertainty}, while \CTRL{} paired with \SEL{} performs only slightly worse. We also experiment with partial observability of voltages (denoted ``p.o.'') by withholding observations of voltages from certain buses, including buses 19 and 31 whose voltage profiles are plotted in (a) and (b). As expected, our controller performs worse with partially observations compared to full observations, but only marginally so. In (c), it is visible that the lack of voltage observation at these 7 nodes does not significantly affect the ability of \SEL{} to learn consistent models.}
    \label{fig:large_uncertainty}
\end{figure*}

\paragraph{Experiment 1: Moderate Uncertainty.} Our first set of experiments runs the robust online voltage controller (\Cref{alg:robust_online_volt_control}, denoted \CTRL$+$\SEL) in a scenario where the norm of the initial parameter estimate is within 50\% of the true parameter norm. \Cref{fig:med_uncertainty} shows that the proposed controller can consistently maintain a voltage profile within the nominal operation range. The solid orange line and shaded orange region represent the mean and $\pm 1$ standard deviation of voltages from the proposed controller, over 4 random choices of $\hat{X}_1$. In contrast, the robust controller \CTRL{} on its own (without the consistent model chasing algorithm \SEL{} to update the model) exhibits large voltage oscillations shown in blue and fails to stabilize the voltage. This is a consequence of \CTRL{} choosing actions based on incorrect information about the network topology. Using the wrong knowledge is even worse than not having a controller (\Cref{fig:baseline_noctrl}) in this case. This demonstrates a clear need for learning the topology rather than purely replying on a robust voltage controller in the case of uncertainty about the network topology. 

\Cref{fig:med_uncertainty_cmc} shows the evolution of the convex model chasing algorithm results across different choices of $\hat{X}_1$. Notably, even though we only approximate the true consistent set through a small random sample, the distance between the learned model $\hat{X}_t$ and true model $X^\star$ decreases monotonically over time. However, the uncertainty does not converge to zero, illustrating that \Cref{alg:robust_online_volt_control} does not perform complete system identification and instead learns just enough about the topology in order to stabilize the system. This is a key novelty of our approach, and enables the algorithm to quickly adapt to uncertain network conditions.

Finally, we compare the performance of \Cref{alg:robust_online_volt_control} with that of the robust controller that has perfect knowledge of the network topology (\Cref{fig:baseline_perfect}). The performance of \Cref{alg:robust_online_volt_control} is comparable despite its lack of knowledge of the topology. One reason for the near-optimal performance of \Cref{alg:robust_online_volt_control} in this case is that the robust controller \CTRL{} turns out to be even more robust than our theoretical results suggest. In particular, empirically it is robust to $\hat{X}$ up to a distance $0.4 \norm{X^\star}_\triangle \approx 66$ away from the true $X^\star$, even though our theoretical results only guarantee its robustness for distances up to $\rho \approx 0.014$.

\paragraph{Experiment 2: Large Uncertainty}
Next, we test the proposed robust online voltage controller in a more challenging setting where there is a large amount of uncertainty. Here, the initial $\hat{X}$ is generated from impedance values with up to $100\%$ error from the true values. Consistent with the moderate uncertainty case, our method manages to maintain voltage stability across all buses. Even though the initial model estimation can be very different compared to the ground truth system (\textit{i.e.}, $\norm{\hat{X} - X^{\star}}_\triangle > 100$), \Cref{alg:robust_online_volt_control} quickly learns from the mistakes and refines the model estimation. Notice in \Cref{fig:large_uncertainty_cmc} how the model error drops quickly during the voltage violation periods between time 10:00 and 12:00. However, the robust controller on its own (without \SEL) fails to control the voltage in the large uncertainty scenario: the voltage profiles at buses 19 (\Cref{fig:large_uncertainty_bus19}) and 31 (\Cref{fig:large_uncertainty_bus31}) deviate significantly from the nominal value. This may lead to unsafe operating conditions that violate regulatory requirements, with potentially catastrophic consequences, \textit{e.g.}, blackouts. In addition, the robust controller on its own is quite sensitive to the initial choice of $\hat{X}$, with a much larger standard deviation than \Cref{alg:robust_online_volt_control}. Similar to the previous experiment, the model error shown in \Cref{fig:large_uncertainty_cmc} decreases monotonically with a nonzero final estimation error.

In this large uncertainty setting, we also experiment with partial observability of voltages (denoted ``p.o.'' in \Cref{fig:large_uncertainty}) by withholding observations of voltages from buses $i \in \{9, 19, 22, 31, 40, 46, 55\}$. That is, we still permit the robust controller \CTRL{} to control the reactive power injection at these buses, but \SEL{} does not use the voltages $v_i(t)$ from these buses as part of the constraints for the consistent sets. As expected, our controller performs worse in this partially-observed setting, but only marginally so. As shown in \Cref{fig:large_uncertainty_cmc}, the lack of voltage observation at these 7 nodes does not significantly affect the ability of \SEL{} to learn consistent models.

\paragraph{Experiment 3: Moderate Initial Error, Large Uncertainty}
Finally, we test our robust online voltage controller with a large uncertainty set ($\alpha = 1.0$) but moderate initial error ($\hat{X}_1$ projected into $\X_{0.5}$). We found the plots of voltage profiles to be nearly indistinguishable from the moderate uncertainty/moderate error setting (\Cref{fig:med_uncertainty}), so we have omitted inclusion of those plots. This empirical observation matches the intuitive idea that for large uncertainty settings, the observed trajectory data is more informative in the definition of consistent sets than the initial uncertainty set $\X$. This explains why our method performs similarly in both the medium uncertainty/medium error and large uncertainty/large error settings.

\section{Conclusion}
\label{sec:conclusion}

This paper provides the first controller that establishes a finite-mistake guarantee for voltage control in a setting with an unknown grid topology. We showed that traditional voltage control algorithms that rely on knowing network information may cause grid instability when given incorrect information about the network topology; furthermore, decentralized network-agnostic control algorithms may also fail when subject to realistic noise and constraints on control inputs. In contrast, our proposed algorithm is able to learn the grid topology in an online fashion and provably converge to a stable controller. Further, simulated experiments on a 56-bus distribution grid demonstrate the effectiveness of our algorithm in a practical scenario.

As the current algorithm is highly centralized, future works may consider more decentralized approaches to topology-robust voltage control in order to enable faster real-time control.  Ideas from works such as ~\cite{yu2022online} can potentially be adapted.
Further studies may also explore loosening the radial topology assumption to accommodate a wider range of distribution grids. This would be a challenging, but important, extension.  Finally, an interesting algorithmic extension is to consider computationally efficient convex body chasing algorithms with better competitive ratios. Existing methods based on Steiner point \cite{argue2019nearly,bubeck2020chasing} achieve nearly-optimal competitive ratio but are computationally inefficient in high dimension settings such as voltage control, so designing efficient approximate Steiner point algorithms could potentially lead to significant performance improvements.

\begin{acks}
We thank Dimitar Ho for useful conversations as well as the anonymous reviewers for their careful reading of this paper and insightful suggestions. This work was supported by NSF grants CNS-2146814, CPS-2136197, CNS-2106403, NGSDI-2105648, and ECCS-2200692.
\end{acks}

\bibliographystyle{ACM-Reference-Format}
\bibliography{sources}

\end{document}